\documentclass[aps,prl,reprint,twocolumn,superscriptaddress,nofootinbib]{revtex4-2}
\usepackage{lineno}
\usepackage{amsmath}
\usepackage{amsthm}
\usepackage{amssymb}
\usepackage{graphicx}
\usepackage{natbib}
\usepackage{aas}
\usepackage{tensor}
\usepackage{mathtools}
\usepackage{mathrsfs}
\usepackage{hyperref}
\usepackage{xcolor}

\newtheorem{lemma}{Lemma}

\theoremstyle{definition}

\theoremstyle{remark}

\DeclareMathOperator{\supp}{supp}
\DeclareMathOperator{\trace}{Tr}

\renewcommand{\d}[0]{\mathrm{d}}
\renewcommand{\vec}[1]{\mathbf{#1}}

\begin{document}
\title[]{Well-defined equations of motion without constraint of external sources}

\author{Kevin~S.~Croker}
\address{Department of Physics and Astronomy, University of Hawai`i at M\=anoa,  2505 Correa Rd., Honolulu, Hawai`i, 96822}
\email{kcroker@phys.hawaii.edu}

\author{Joel~L.~Weiner}
\address{Department of Mathematics, University of Hawai`i at M\=anoa, 2565 McCarthy Mall, Honolulu, Hawai`i 96822}

\author{Duncan~Farrah}
\address{Institute for Astronomy, University of Hawai`i at M\=anoa,  2680 Woodlawn Dr., Honolulu, Hawai`i, 96822}
\address{Department of Physics and Astronomy, University of Hawai`i at M\=anoa,  2505 Correa Rd., Honolulu, Hawai`i, 96822}

\begin{abstract}
  We present a new approach to constrained classical fields that enables the action formalism to dictate how external sources must enter the resulting equations of motion.
    If symmetries asserted upon the varied fields can be modeled as restrictions in Fourier space, we prove that these restrictions are automatically applied to external sources in an unambiguous way.
  In contrast, the typical procedure inserts symmetric ansatze into the Euler-Lagrange differential equations, even for external sources not being solved.
  This requires \emph{ad hoc} constraint of external sources, which can introduce leading-order errors to model systems despite superficial consistency between model field and source terms.
  To demonstrate, we consider Robertson-Walker cosmologies within General Relativity and prove that 
  the influence of point-like relativistic pressure sources on cosmological dynamics cannot be excluded by theoretical arguments.
\end{abstract}
  
\maketitle
  The principle of stationary action is the theoretical foundation of contemporary physics\cite[][\S2.2]{peskin2018introduction}.
    Unlike differential equations, the integral nature of the action makes it attractive for working with non-local constraints, such as those in Fourier space.
  Such constraints are commonly made to reduce the complexity of model equations.
  For example, a symmetry ansatz that removes coordinate dependence in $\hat{z}$ has vanishing Fourier modes off of the $\widehat{xy}$  plane.
  Traditionally, simplifying assumptions such as these are applied to all fields in a given physical model.
  This includes both the dynamical degrees of freedom within the model and fields representing external sources.
  In applied settings, this is justified: the experimenter has physical control over sources and boundaries.
  In observational settings, however, this is not always the case.
  For example, in cosmology, one uses data to reconstruct the types of source present within the Universe and their distribution in spacetime.
  Though there is little disagreement that General Relativity (GR) is the appropriate framework within which to construct cosmological models, there has been significant debate concerning what simplifying assumptions can be made.
  Many authors have argued that the formation of structures must necessarily be taken into account when constructing a cosmological model\cite[e.g.][]{rasanen2004dark, kolb2006cosmic, ellis2011inhomogeneity, buchert2015there}.
  This ``cosmological backreaction,'' though well-motivated by the non-linear nature of Einstein's field equations, has been disputed\cite[e.g.][]{wetterich2003can,  green2014well, green2016simple, kaiser2017there}.
  The debate over backreaction highlights the ambiguities that can present when determining the appropriate notion of ``source'' in idealized models.

  In this paper, we will derive the Euler-Lagrange equations for classical fields on the $N$-torus in a manner that incorporates the additional symmetries imposed by Fourier-space constraint.
We will discover that symmetries imposed on model fields automatically become applied to other fields held fixed during the variation, i.e. sources.
Consequently, a significant advance of our approach to the Euler-Lagrange equations is an unambiguous procedure for converting an unconstrained, microphysical source into a source appropriate for simplified dynamical models.
This work generalizes an existing technique, applicable only with position-independent model fields, to arbitrary shaping in Fourier space\cite{cw2018part1, crokrunfarr}.

\emph{Definitions:}
\label{sec:definitions}
Let the $N$-torus be denoted by $\mathbb{T}^N$, let $I$ be a closed interval on $\mathbb{R}$, and define
\begin{linenomath}
\begin{align}
  \mathcal{M} := I \times \mathbb{T}^N.
\end{align}
\end{linenomath}
Note that $\mathcal{M}$ can accommodate a flat Lorentz metric $f$ with global coordinates $(\mathbf{x}, \eta)$ and that $\eta = \eta_0$ defines a spatial slice with respect to $f$.
In these coordinates, $f = - d\eta^2 + d\mathbf{x}^2$.
With respect to $f$, $\mathbb{T}^N$ is the product of $N$, mutually orthogonal, circles of equal, but arbitrarily large, length.
In this case, $\vec{x}$ becomes a coordinate on the simply connected covering space of the flat $N$-torus: $\mathbb{R}^N$.
Each component of $\vec{x}$ is then well-defined up to integral multiples (translations) of $\mathcal{V}^{1/N}$ on any $\mathbb{T}^N$, where $\mathcal{V}$ is the volume of the torus.

Let $A$ be one of a finite collection of fields on $\mathcal{M}$.
Let $\mathcal{L}(A, \partial_\mu A, \dots)$ be a Lagrange density with convergent power series
in $A$, these other fields, and their derivatives.
We suppose that $A$ can be represented by its spatial Fourier transform $A(\vec{k}, \eta)$ on $\mathbb{T}^N$,
\begin{linenomath}\begin{align}
  A(\vec{x}, \eta) = \sum_\vec{k} e^{i\vec{k}\cdot\vec{x}} A(\vec{k}, \eta). 
\end{align}\end{linenomath}
We will often suppress coordinate dependence for concision.
All modes of the $N$-torus are integer multiples of a fundamental mode,
\begin{linenomath}\begin{align}
  \mathcal{K} := \frac{2\pi}{\mathcal{V}^{1/N}}.
\end{align}\end{linenomath}
The action $S$ characterizing the dynamical evolution of $A$ is the integral of $\mathcal{L}$ over $\mathcal{M}$
\begin{linenomath}\begin{align}
    S := \int_\mathcal{M} \mathcal{L}~\d^{N+1} x.
\end{align}\end{linenomath}
The equations of motion for the field $A$ are determined by demanding that $S$ be critical to all $C^\infty$ variations $\delta A$, of compact support in space and time, \emph{which satisfy the same constraints as the field $A$}\cite[][\S5]{LanczosVariational}.
We emphasize that $\delta A$ is an entirely new field temporarily introduced to determine equations of motion.

Evaluation of the varied Lagrange density gives, by definition
\begin{linenomath}\begin{align}
  \int_\mathcal{M}  \delta A\left\{\frac{\delta \mathcal{L}}{\delta A}\right\}~\d^{N+1}x := 0. \label{eqn:nascent_eom}
\end{align}\end{linenomath}
If the degree of freedom $A$ were unconstrained, then the equations of motion $\delta \mathcal{L}/\delta A = 0$ would follow from the Fundamental Lemma of Variational Calculus.
The Fundamental Lemma \cite[e.g.][\S2.2]{giaquinta2004calculus}, however, requires that the variations $\delta A$ include all $C^\infty$ functions of compact support.
If the field $A$ is constrained in Fourier space, the variations $\delta A$ are necessarily constrained in the same way.
This means that the variations $\delta A$ are not sufficient to allow application of the Fundamental Lemma to extract equations of motion.

\emph{Results:}
We will now generalize the Fundamental Lemma to the case where degrees of freedom are constrained in $k$-space.
\begin{lemma}
  \label{lem:core}
  Let $A(\vec{k}, \eta)$ be the Fourier transform of some field $A(\vec{x}, \eta)$ that appears in $\mathcal{L}$.
  Let $V$ denote the support of $A(\vec{k}, \eta)$.
  Then the equations of motion for $A$ are
  \begin{align}
    \frac{\delta \mathcal{L}}{\delta A} * \mathcal{F}^{-1}\left[\mathbf{1}_V\right] = 0, \label{eqn:final_result}
  \end{align}
  where $*$ denotes convolution\footnote{
    convolution on $\mathbb{T}^N$ is understood as circular/cyclic convolution, see \cite[][p.~265]{bracewell1986fourier}}, $\mathcal{F}^{-1}$ denotes the inverse Fourier transform, and $\mathbf{1}$ denotes the indicator function.
\end{lemma}
\begin{proof}
  Suppose we have varied $\mathcal{L}$ with respect to $A$.
  Use compactness of $\mathcal{M}$ to write the varied action as an iterated integral
  \begin{linenomath}\begin{align}
      \int_\eta \int_\vec{x}  \delta A\left\{\frac{\delta \mathcal{L}}{\delta A}\right\}\d\vec{x}~\d\eta. \label{eqn:fubini}
  \end{align}\end{linenomath}
  The temporal integrand can be recognized as an inner product on $L^2$ (i.e. the Lebesgue square-integrable functions)
  \begin{linenomath}\begin{align}
      \int_\eta \left<\delta A, \frac{\delta \mathcal{L}}{\delta A}\right>_{\mathbb{T}^N}~\d\eta. \label{eqn:l2_ip}
  \end{align}\end{linenomath}
  We compute this inner product in $k$-space.
  Use the indicator function to express the Fourier transform of a variation as
  \begin{linenomath}\begin{align}
      \delta A(\vec{k}, \eta) := \overline{\delta A}(\vec{k}, \eta) \mathbf{1}_{V}(\vec{k}), \label{eqn:fourier_variation}
  \end{align}\end{linenomath}
  where $\overline{\delta A}(\vec{x}, \eta)$ denotes an arbitrary $C^\infty$ variation of compact support in position and time.
  The inner product in Expression~(\ref{eqn:l2_ip}) is $L^2$ basis invariant (i.e. Plancheral's theorem):
  \begin{linenomath}\begin{align}
      \left<\delta A, \frac{\delta \mathcal{L}}{\delta A}\right>_{\vec{x}} = \left<\overline{\delta A}~\mathbf{1}_{V}, \frac{\delta \mathcal{L}}{\delta A}\right>_{\vec{k}}.
  \end{align}\end{linenomath}
  Using standard properties of the inner product on $L^2$, we may commute the indicator function
  \begin{linenomath}\begin{align}
      \left<\overline{\delta A}~\mathbf{1}_V, \frac{\delta \mathcal{L}}{\delta A}\right>_{\vec{k}} = \left<\overline{\delta A}, \frac{\delta \mathcal{L}}{\delta A} \mathbf{1}_V \right>_{\vec{k}}.
  \end{align}\end{linenomath}
  Using Plancheral's theorem to return to position-space produces the convolution of $\delta \mathcal{L}/\delta A$ against the inverse Fourier transform of the indicator function for $V$ 
  \begin{linenomath}\begin{align}
      \left<\overline{\delta A}, \frac{\delta \mathcal{L}}{\delta A} \mathbf{1}_V \right>_{\vec{k}} = \left<\overline{\delta A}, \frac{\delta \mathcal{L}}{\delta A} * \mathcal{F}^{-1}\left[\mathbf{1}_V\right] \right>_{\vec{x}}.
  \end{align}\end{linenomath}
  We conclude that the varied action is equal to
  \begin{linenomath}\begin{align}
      \int_\mathcal{M} \overline{\delta A} \left\{\frac{\delta \mathcal{L}}{\delta A} * \mathcal{F}^{-1}\left[\mathbf{1}_V\right]\right\}\d^{N+1}x~ := 0. \label{eqn:eevee_convolution}
  \end{align}\end{linenomath}
  The variations under the action integral $\overline{\delta A}$ are now arbitrary in $C^\infty$ and have compact support.
  In this form, the Fundamental Lemma of Variational Calculus can be applied to extract consistent equations of motion from under the action.\end{proof}

For clarity, we note that the consistent equations of motion, given in Equation~(\ref{eqn:final_result}), expand to
\begin{align}
  \left\{ \partial_\mu \left[\frac{\partial \mathcal{L}}{\partial\left(\partial_\mu A\right)}\right] - \frac{\partial \mathcal{L}}{\partial A} + \cdots \right\} * \mathcal{F}^{-1}\left[\mathbf{1}_V\right] = 0, \label{eqn:final_result_expanded}
\end{align}
where dots denote higher derivative terms that may result from more intricate $\mathcal{L}$.
Note that the typical Euler-Lagrange equations
\begin{align}
  \partial_\mu \left[\frac{\partial \mathcal{L}}{\partial\left(\partial_\mu A\right)}\right] - \frac{\partial \mathcal{L}}{\partial A} + \cdots = 0
\end{align}
obtain from Equation~(\ref{eqn:final_result_expanded}) only when $V \to \mathbb{R}^n$.
In this limit, $\mathcal{F}^{-1}\left[\mathbf{1}_V\right]$ becomes the Dirac delta distribution, which is the identity under convolution.
Explicitly, Equation~(\ref{eqn:final_result_expanded}) is the correct generalization of the Euler-Lagrange equations in the presence of a particular class of \emph{non-local} constraint.
This class of non-local constraint is ubiquitous.
For example, whenever symmetry assumptions asserted upon the field $A$ remove coordinate dependence, Fourier-space support along that reciprocal coordinate axis collapses to the origin.
We emphasize that the convolution in Equation~(\ref{eqn:final_result_expanded}) is non-trivial.
If $\mathcal{L}$ contains external sources, \emph{a priori} constraint of external sources to respect the symmetries of $A$ is no longer required.
Additional constraint of external sources, distinct from that required by Equation~(\ref{eqn:final_result_expanded}), will at best introduce additional physical assumptions into the model.
At worst, \emph{ad hoc} constraint will fail to correctly capture the symmetries required by the action, yet still introduce a source appearing superficially to do so.

Note that it is possible to replace $\mathbf{1}_V$ in Equation~(\ref{eqn:fourier_variation}) with any other strictly positive function $\mu(\vec{k})$ with support on $V$.
Because both source and fields are filtered in Equation~(\ref{eqn:eevee_convolution}), the true $A$ can still be recovered by a deconvolution against $\mathcal{F}^{-1}\left[\mu^{-1}\right]$.

\emph{Discussion:}
Lemma~\ref{lem:core} allows, for the first time, unambiguous modeling of physical systems where the source terms are not under the modeler's direct experimental control.
  Under Lemma 1, various assumptions about the external source translate into testable observational consequences.
  To demonstrate, we will apply Lemma 1 to the well-known Robertson-Walker (RW) system.
In GR, the field to be determined is the spacetime metric tensor $g$.
In order that cosmological problems remain well-posed\cite{d1976existence}, attention is often restricted\cite{peebles1993principles} to $g$ defined on $I \times \mathbb{T}^3$. 
The source for $g$ is the stress-energy tensor $T$.
It is not varied when determining the Einstein field equations.
Instead, the stress is defined implicitly through the variation of some matter action $S_M$ with respect to $g$,
\begin{align}
  \delta S_M \propto: \int_\mathcal{M} T \cdot \delta g\sqrt{-\det{g}}~\d^4x,
\end{align}
where $\cdot$ denotes complete contraction.
Note that $T$ is only well-defined under a varied action integral.
Variation of the remaining contribution to the GR Lagrange density gives,
\begin{align}
  \delta S_{\mathrm{GR}} = \int_{\mathcal{M}} \delta g\cdot\left[G - \kappa^2T\right]\sqrt{-\det{g}}~\d^4x,
\end{align}
where $G$ denotes the Einstein tensor and $\kappa^2$ is a dimensionful coupling constant.

If $g$ is unconstrained, the Fundamental Lemma can be applied and the Einstein field equations can be extracted from under the varied integral.
In RW cosmology, however, $g$ is defined to be isotropic and homogeneous, e.g.
\begin{align}
  g := a^2(\eta) f, 
\end{align}
where recall that we have defined $f$ to be the flat metric on $\mathcal{M}$.
In this form, notice that $g$ is constrained in Fourier space: all of its spatial derivatives vanish.
In other words, $g$ and its variations $\delta g$ are constrained to have singleton support in $k$-space
\begin{align}
  \supp g(\vec{k}) = \supp \delta g(\vec{k}) = \vec{0}.
\end{align}
As proved in Lemma~\ref{lem:core}, equations of motion (i.e. Einstein's equations appropriate for the RW ansatz) cannot be extracted without first convolving against $\mathcal{F}^{-1}[\delta(\vec{k})]$.
Performing\cite{cw2018part1} this convolution and extracting the equation of motion gives
\begin{align}
  6\mathcal{V}\frac{\d^2a}{\d \eta^2} - \kappa^2a^3\int_\mathcal{V}\trace{T}(\eta,\vec{x})~\d^3x = 0,
\end{align}
where $\trace$ denotes the (coordinate invariant) trace of the stress.
Assuming the stress to be of Type I, rearranging and dividing by the volume gives
\begin{align}
  \frac{\d^2a}{\d \eta^2} = \frac{4\pi G}{3}a^3\left<\overline{\rho} - \sum_{i=1}^3 \overline{\mathcal{P}}_i\right>_\mathcal{V},
\end{align}
where Newton's constant $G$ enters via $\kappa^2 = 8\pi G$.
In this equation, $\overline{\rho}$ is the time-like eigenvalue (i.e. energy density) and $\overline{\mathcal{P}}_i$ are the space-like eigenvalues (i.e. principal pressures)\cite[][\S B]{crokrunfarr}.
If one writes each principal pressure as an isotropic pressure $\mathcal{P}(\eta, \vec{x})$ plus an anisotropic contribution, the anisotropic contributions will vanish when averaged over the 3-torus, leaving
\begin{align}
  \frac{\d^2a}{\d \eta^2} = \frac{4\pi G}{3}a^3\left<\overline{\rho} - 3\mathcal{P}\right>_\mathcal{V}. \label{eqn:friedmann}
\end{align}
This Friedmann equation has the expected form, but the position-independent energy density and isotropic pressure are unambiguously volume averages over the position-dependent, microphysical, quantities.
This result contradicts widely repeated arguments that stresses interior to compact objects do not contribute to the Friedmann source\citep[e.g.][]{einstein1945influence, einstein1946corrections, weinberg2008cosmology, peebles2020large}.
  We have proved that not only can such stresses contribute, but they can influence dynamics at leading order.
The existence of objects that contribute leading-order pressures to Friedmann's equations now becomes an observational question.

In summary, we have leveraged the action integral itself to derive the appropriate Euler-Lagrange equations in the presence of Fourier-space constraint.
  Our result applies generally to classical fields on the $N$-torus and is relevant because Fourier-space constraint is implicit whenever symmetry assumptions remove coordinate dependence.
  The resulting Euler-Lagrange equations feature a necessary convolution operation, which reduces to the familiar presentation in the unconstrained limit.
  This convolution operation interacts with unconstrained external sources to produce effective sources that are consistent with the model constraint.
  The need for \emph{ad hoc} assumptions to align external sources with model symmetries is replaced with an unambiguous procedure, determined directly from the action.

\bibliographystyle{unsrtnat}
\bibliography{level26}

\end{document}